\documentclass[fleqn,12pt]{article}
\usepackage{amsmath,amssymb,epsfig,times,graphics}
\usepackage{amssymb}
\usepackage{amsmath}
\usepackage{CJK}
\usepackage{graphics}
\usepackage{graphicx}
\usepackage{dcolumn}
\usepackage{bm}
\usepackage{subfigure}
\usepackage{lineno}
\usepackage{amsthm}
\begin{document}

\title{\bf The overshoot and phenotypic equilibrium in characterizing cancer dynamics of reversible phenotypic plasticity}

\date{}
\maketitle

\author{Xiufang Chen$^{7,2}$, Yue Wang$^{3}$, Tianquan Feng$^{4}$, Ming Yi$^{5,6}$, Xingan Zhang$^{2,*}$, Da Zhou$^{1,*}$}

\begin{enumerate}
  \item School of Mathematical Sciences, Xiamen University,
Xiamen 361005, P.R. China \\ ($*$Co-corresponding Author, zhouda@xmu.edu.cn)
  \item School of Mathematics and Statistics, Central China Normal University,
Wuhan 430079, P. R. China \\ ($*$Co-corresponding Author, zhangxinan@hotmail.com)
  \item Department of Applied Mathematics, University of Washington,
Seattle, WA 98195, USA
  \item School of Teachers' Education, Nanjing Normal University, Nanjing 210023, China
  \item Department of Physics, College of Science, Huazhong Agricultural University, Wuhan, Hubei 430070, China
  \item Key Laboratory of Magnetic Resonance in Biological Systems, Wuhan Institute of
Physics and Mathematics, Chinese Academy of Sciences, Wuhan, 430071, P. R. China
  \item School of Computer Science and Information Engineering, Qilu Institute of Technology, Jinan, Shandong 250000, China
\end{enumerate}

\newpage

\begin{abstract}
The paradigm of phenotypic plasticity indicates reversible relations of different cancer cell phenotypes,
which extends the cellular hierarchy proposed by the classical cancer stem cell (CSC) theory.
Since it is still questionable if the phenotypic plasticity is a crucial improvement to the hierarchical model
or just a minor extension to it,
it is worthwhile to explore the dynamic behavior characterizing the reversible phenotypic plasticity.
In this study we compare the hierarchical model and the reversible model in predicting the cell-state
dynamics observed in biological experiments.
Our results show that the hierarchical model shows significant disadvantages over the reversible model in
describing both long-term stability (phenotypic equilibrium) and short-term transient dynamics (overshoot) of cancer cells.
In a very specific case in which the total growth of population due to each cell type is identical,
the hierarchical model predicts neither phenotypic equilibrium nor overshoot, whereas the reversible model
succeeds in predicting both of them. Even though the performance of the hierarchical model can be improved
by relaxing the specific assumption, its prediction to the phenotypic equilibrium strongly depends on a precondition that may be
unrealistic in biological experiments, and it also fails to capture the overshoot of CSCs. By comparison, it is more likely for
the reversible model to correctly describe the stability of the phenotypic mixture and various types of overshoot
behavior.

\end{abstract}

\section{Introduction}
\label{Introduction}

The cancer stem cell theory has provided a hierarchical
model of how diverse cancer cells being organized
\cite{reya2001stem,jordan2006cancer,dalerba2007cancer}.
Similar to the stem cell theory in normal tissues,
this hierarchical model assumes that a small number of stem-like cancer cells
(termed cancer stem cells, CSCs) are capable of self-renewal
and differentiation into other more committed cancer cells
(termed non-stem cancer cells, NSCCs) but not vice versa.
That is, CSCs are thought to be at the apex of this cellular
hierarchy. However, some recent researches may
extend this unidirectional relation of cancer cells.
It has been reported that cancer cells can convert from
NSCC phenotype to CSC phenotype
(\emph{e.g.} breast cancer \cite{meyer2009dynamic,chaffer2013poised},
melanoma \cite{quintana2010phenotypic}, colon cancer \cite{yang2012dynamic}
and glioblastoma multiforme \cite{fessler2015endothelial}).
Furthermore, the interconversions among cancer cell phenotypes
have also been found (breast cancer \cite{gupta2011stochastic}).
All these works indicate reversible relations of different cancer cells
(termed \emph{phenotypic plasticity}, which has long been an issue
concerned in bacterial populations \cite{kussell2005phenotypic}).

Very recently special attention has already been paid to the reversible cancer models
with the phenotypic plasticity by theoreticians.
In particular, dos Santos and da Silva explained the variable frequencies
of CSCs in tumors by establishing a model
with stochastic cell plasticity \cite{dos2013possible,dos2013noise}.
Leder \emph{et al} investigated the model of reversible conversions between
the stem-like resistant cells (SLRCs) and the differentiated
sensitive cells (DSCs) in glioblastomas  \cite{leder2014mathematical}.
Wang \emph{et al} showed how tumor heterogeneity arises in the model of
cooperating CSC hierarchy with cell plasticity \cite{wang2014dynamics}.
Zhou \emph{et al} showed that the de-differentiation from NSCC phenotype to CSC phenotype
was essential for explaining the transient increase of the minority populations of CSCs
observed in cancer cell lines \cite{zhoup2013opulation,zhou2014multi}.
Chen  \emph{et al} studied stochastic models that capture the transitions
between endocrine therapy responsive and resistant states of breast cancer cells
\cite{chen2014mathematical}. Zhou \emph{et al} investigated nonequilibrium dynamics
with phenotype transitions of cancer cells \cite{zhou2014nonequilibrium}.
Jilkine and Gutenkunst studied the effect of de-differentiation on time to
mutation acquisition in cancers \cite{jilkine2014effect}.

\begin{figure}
\begin{center}
\includegraphics[width=1.2\textwidth]{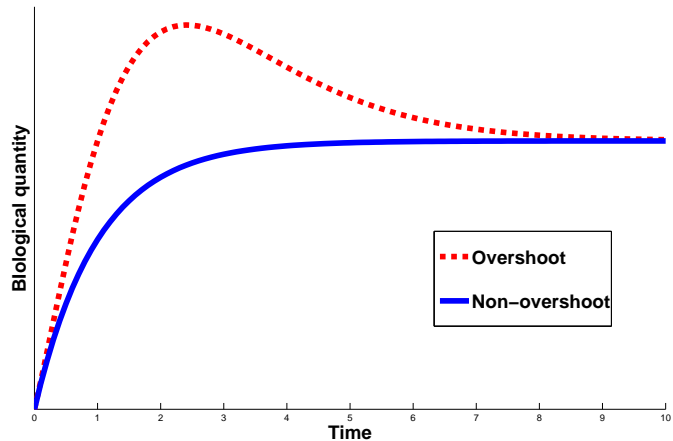}
\caption{Overshoot is a type of non-monotonic phenomenon \cite{jia2013overshoot} that,
starting from a state that is lower than the equilibrium level,
the process first increases above the final equilibrium
level and then gradually decreases to it.}
\end{center}
\end{figure}

However, it is still questionable if the phenotypic plasticity is a crucial improvement
to the hierarchical model or just a minor extension to it \cite{zapperi2012cancer,easwaran2014cancer}.
Thus a rigorous analysis on the characteristics owned by the model
with the phenotypic plasticity is necessary for model validation.
In this study, we try to investigate this issue by giving a comparative study of the reversible
model and the hierarchical model.
Note that Gupta \emph{et al} studied SUM149 and SUM159 breast cancer cell lines
\cite{gupta2011stochastic}, in which
two interesting phenomena were observed: \emph{phenotypic equilibrium}
and \emph{overshoot}. For the phenotypic equilibrium, they found that
the breast cancer cell lines will tend to a stable phenotypic mixture
over time regardless of initial states.
Similar results on the phenotypic equilibrium were also reported in
\cite{chang2008transcriptome,chaffer2011normal,yang2012dynamic}.
For the overshoot, they found that shortly after the cell sorting,
the proportion of the minority subpopulation
will increase transiently above the final equilibrium level, and then
decrease to it (Fig. 1). Enlightened by these observations,
we show that the reversible model is more capable of
capturing the phenotypic equilibrium and overshoot than the hierarchical model.
Under a very specific assumption that the total population growth due to each cell
phenotype is identical, the hierarchical model can perform neither phenotypic
equilibrium nor overshoot, whereas the reversible model succeeds in predicting both
the phenotypic equilibrium and three types of overshoots
(asynchronous, synchronous and oscillating overshoots). Even though the performance of
the hierarchical model can be improved by relaxing the specific assumption,
it is still not good enough to correctly capture the cell-state dynamics.
On one hand, the phenotypic equilibrium predicted by the hierarchical
model strongly depends on the condition that the self-contributed
growth rate by CSCs is faster than that of more committed cancer cells.
However, it has been reported that this condition cannot be satisfied in some cancers
\cite{patrawala2005side,fillmore2008human}. On the other hand, by fitting the two models
to experimental data in \cite{gupta2011stochastic}, the hierarchical
model can only fit the overshoot of NSCCs, but it cannot capture the overshoot
of CSCs. By contrast, the reversible model can fit both of them.
Therefore, our results generally imply that the reversible model
shows distinct advantages over the hierarchical model in
predicting both long-term and transient dynamics of cancer cells.

The paper is organized as follows. The mathematical model is presented in Section 2.
Main results are shown in Section 3, where we give a comparative study
of the reversible and hierarchical models. Conclusions are presented in Section 4.

\section{Model}
\label{Model}

In this section we describe the assumptions of our model.
We model cancer as population dynamics of cancer cells, each
cancer cell can be assigned to one of the following cell phenotypes:
$\textrm{CSC}_0$,
$\textrm{NSCC}_1$,
$\textrm{NSCC}_2$, ..., and
$\textrm{NSCC}_n$.
Let $X_0(t)$, $X_1(t)$,..., $X_n(t)$ be the cell numbers of $\textrm{CSC}_0$, $\textrm{NSCC}_1$, ..., $\textrm{NSCC}_n$ at time $t$
respectively, then the model can be described as a cellular dynamical system of $\overrightarrow{X_t}=(X_0(t), X_1(t),..., X_n(t))^T$.
In particular, the model will reduce to a two-phenotypic model when $n=1$. That is, besides CSC-phenotype, all the other
cancer cells are grouped into one whole NSCC-phenotype. This CSC-NSCC model has extensively been investigated in
previous literature \cite{zapperi2012cancer,dos2013possible,dos2013noise,zhoup2013opulation,leder2014mathematical,wang2014dynamics}.
Here we pay more attention to the multi-phenotypic case ($n\geq 2$). As a starting point,
the case with $n=2$ should be a natural and favorable candidate for theoretical analysis.
Meanwhile, motivated by Gupta \emph{et al} \cite{gupta2011stochastic} where three phenotypes (stem-like, basal and luminal) were
identified, we are concerned about $\textrm{CSC}_0$-$\textrm{NSCC}_1$-$\textrm{NSCC}_2$ model in light of both theoretical and experimental reasons.
Unlike the three-phenotypic cell lineages investigated in \cite{johnston2007mathematical,liu2013nonlinear, pei2015fluctuation}
where cancer stem cells, progenitor cells, and terminally differentiated cells are hierarchically cascaded,
here we are more interested in the cell population structure investigated in \cite{gupta2011stochastic}.
In their work, CSCs can differentiate into $\textrm{NSCC}_1$ and $\textrm{NSCC}_2$, respectively.
They were not assumed to be cascaded. Moreover, by the phenotypic plasticity,
they can interconvert into each other and de-differentiate into CSCs.

We now present the cellular processes included in our model.
For $\textrm{CSC}_0$, it can not only divide symmetrically into two identical
$\textrm{CSC}_0$ daughters, but also divide asymmetrically into $\textrm{CSC}_0$ and
$\textrm{NSCC}_i$ ($i=1$ or $2$)
\cite{boman2007symmetric,dingli2007symmetric,marciniak2009modeling}.
That is,
\begin{itemize}
  \item  $CSC_0\overset{\alpha_1}{\longrightarrow} CSC_0+CSC_0$.
  \item  $CSC_0\overset{\alpha_2}{\longrightarrow} CSC_0+NSCC_1$.
  \item  $CSC_0\overset{\alpha_3}{\longrightarrow} CSC_0+NSCC_2$
  \footnote{It should be pointed out that, the symmetric differentiation ``$CSC_0{\longrightarrow} NSCC_i+NSCC_i$'' is also an important type of transition \cite{morrison2006asymmetric} that should be considered in the model. However, it can be shown that the main results we are interested in this work can still hold (with only minor changes in proofs) even if adding this transition into Eq. (\ref{Linear1}). Therefore, to keep the model as minimal as possible, the symmetric differentiation is not included in presented model.}.
\end{itemize}
For $\textrm{NSCC}_1$ and $\textrm{NSCC}_2$, besides the phenotypic plasticity, they can also perform symmetric divisions and cell death. That is,
\begin{itemize}
  \item  $NSCC_1\overset{\beta_1}{\longrightarrow} NSCC_1+NSCC_1$;
  \item  $NSCC_1\overset{\beta_2}{\longrightarrow} CSC_0$;
  \item  $NSCC_1\overset{\beta_3}{\longrightarrow} NSCC_2$;
  \item  $NSCC_1\overset{\beta_4}{\longrightarrow} \emptyset$;
  \item  $NSCC_2\overset{\gamma_1}{\longrightarrow} NSCC_2+NSCC_2$;
  \item  $NSCC_2\overset{\gamma_2}{\longrightarrow} CSC_0$;
  \item  $NSCC_2\overset{\gamma_3}{\longrightarrow} NSCC_1$;
  \item  $NSCC_2\overset{\gamma_4}{\longrightarrow} \emptyset$.
\end{itemize}
Based on the above cellular processes, the dynamics of  $\overrightarrow{X_t}=(X_0, X_1, X_2)^T$
can be captured by the following ordinary differential equations (ODEs)
\begin{linenomath*}
\begin{equation}
 \left\{
   \begin{aligned}
   \frac{dX_0}{dt} &= \alpha_1 X_0+\beta_2 X_1+\gamma_2 X_2\\
   \frac{dX_1}{dt} &= \alpha_2 X_0+(\beta_1-\beta_2-\beta_3-\beta_4) X_1+\gamma_3 X_2\\
   \frac{dX_2}{dt} &= \alpha_3 X_0+\beta_3 X_1+(\gamma_1-\gamma_2-\gamma_3-\gamma_4)X_2\\
    \end{aligned}
   \right.
   \label{Linear1}
\end{equation}
\end{linenomath*}
By letting
\begin{linenomath*}
\begin{equation}
Q=[q_{ij}]=\left(\begin{array}{ccc}
\alpha_1 & \beta_2 & \gamma_2  \\
\alpha_2 & \beta_1-\beta_2-\beta_3-\beta_4 & \gamma_3  \\
\alpha_3 & \beta_3 & \gamma_1-\gamma_2-\gamma_3-\gamma_4  \\
\end{array}\right),
\label{Matrix1}
\end{equation}
\end{linenomath*}
Eq. (\ref{Linear1}) can be expressed as
\begin{linenomath*}
\begin{eqnarray}
\frac{d \overrightarrow{X_t}}{d t}=Q\overrightarrow{X_t}.
\label{Linear2}
\end{eqnarray}
\end{linenomath*}
For this model, we have the following remarks:\\
\emph{(1)}
Each element $q_{ij}$ of $Q$ can be seen as the per
capita growth rate of $i^{th}$ phenotype contributed by
$j^{th}$ phenotype. Let us take the first row of $Q$ as an example.
$\alpha_1$ is the symmetric division rate of $\textrm{CSC}_0$, thus
the growth rate of $X_0$ due to $\textrm{CSC}_0$ themselves can be expressed
as $\alpha_1 X_0$; $\beta_2$ and $\gamma_2$ are the
de-differentiation rates
from $\textrm{NSCC}_1$ and $\textrm{NSCC}_2$ respectively,
so the growth rate of
$X_0$ due to $\textrm{NSCC}_1$ and $\textrm{NSCC}_2$ is $\beta_2 X_1+\gamma_2 X_2$.
In this way, it is easy to explain the biological meaning of
``$dX_0/dt = \alpha_1 X_0+\beta_2 X_1+\gamma_2 X_2$'', \emph{i.e.} the rate of change of
$X_0$ is the sum of the growth rates contributed by all the three cell phenotypes in the
population. Similarly, we can explain the other two equations in Eq. (\ref{Linear1}).\\
\emph{(2)}
Note that $\beta_2$, $\beta_3$, $\gamma_2$ and $\gamma_3$ are the parameters associated with the phenotypic plasticity,
the model will reduce to the hierarchical model by letting them be zero. Accordingly,
$Q$ will become to a lower-triangular matrix
\begin{linenomath*}
\begin{equation}
Q^*=[q^*_{ij}]=\left(\begin{array}{ccc}
\alpha_1 & 0 & 0  \\
\alpha_2 & \beta_1-\beta_4 & 0  \\
\alpha_3 & 0 & \gamma_1-\gamma_4  \\
\end{array}\right).
\label{Matrix2}
\end{equation}
\end{linenomath*}
\emph{(3)}
Let $W=X_0+X_1+X_2$ be the total number of the population, then
\begin{linenomath*}
\begin{equation}
\frac{dW}{dt} = (\alpha_1+\alpha_2+\alpha_3) X_0+(\beta_1-\beta_4) X_1+(\gamma_1-\gamma_4) X_2.
\label{total}
\end{equation}
\end{linenomath*}
Note that the coefficient of each $X_i$ on right-hand side of Eq. (\ref{total}) is just the
the sum of the corresponding column in $Q$, \emph{i.e.}
the per capita growth rate of the whole population due to
$X_i$. In particular, when
\begin{linenomath*}
\begin{equation}
\alpha_1+\alpha_2+\alpha_3=\beta_1-\beta_4=\gamma_1-\gamma_4=\kappa,
\end{equation}
\end{linenomath*}
the growth rate of the whole population due to each phenotype
will be identical. This special case implies meaningful biological significance
and has been reported in some cancer cell lines (breast cancer \cite{gupta2011stochastic}
and colon cancer \cite{zhoup2013opulation}). In this case
Eq. (\ref{total}) will reduce to a very simple form
\begin{linenomath*}
$$\frac{dW}{dt} = \kappa W,$$
\end{linenomath*}
that is, $W$ will grow exponentially with constant rate. We will show that,
even in this special case the reversible model and the hierarchical models are quite different in
predicting the phenotypic equilibrium and overshoot. For convenience, we denote
\begin{linenomath*}
\begin{equation}
A_1=\alpha_1+\alpha_2+\alpha_3,
\end{equation}
\begin{equation}
A_2=\beta_1-\beta_4,
\end{equation}
\begin{equation}
A_3=\gamma_1-\gamma_4.
\end{equation}
\end{linenomath*}
\emph{(4)}
Note that all the parameters are assumed to be constant.
In other words, there is no feedback control included in presented model.
This is of course a simplification of the biological facts. To model real
biological systems, feedback mechanisms are indispensable
\cite{lander2009cell,lo2009feedback,komarova2013principles,liu2013nonlinear}.
However, we will show that our model is not overly simplistic to characterize
the phenotypic plasticity in comparison to the cellular hierarchy.

In the remaining of this section, we will convert the cell number equations Eq. (\ref{Linear1})
into the following cell proportion equations Eq. (\ref{Nonlinear1}).
This is because in reality, relative cell numbers, rather than absolute cell numbers,
are usually measured by fluorescence-activated cell sorting (FACS) experiments.
Let $x_0=X_0/W$, $x_1=X_1/W$ and $x_2=X_2/W$, by variable substitutions in Eq. (\ref{Linear1})
we have
\begin{linenomath*}
\begin{equation}
 \left\{
   \begin{aligned}
   \frac{dx_0}{dt} &= \alpha_1 x_0+\beta_2 x_1+\gamma_2 x_2-x_0(A_1x_0+A_2x_1+A_3x_2)\\
   \frac{dx_1}{dt} &= \alpha_2 x_0+(\beta_1-\beta_2-\beta_3-\beta_4) x_1+\gamma_3 x_2-x_1(A_1x_0+A_2x_1+A_3x_2)\\
   \frac{dx_2}{dt} &= \alpha_3 x_0+\beta_3 x_1+(\gamma_1-\gamma_2-\gamma_3-\gamma_4)x_2-x_2(A_1x_0+A_2x_1+A_3x_2)\\
    \end{aligned}
   \right.
   \label{Nonlinear1}
\end{equation}
\end{linenomath*}
Note that $x_0+x_1+x_2=1$, the freedom of Eq. (\ref{Nonlinear1}) is actually two. By letting
$x_2=1-x_0-x_1$, Eq. (\ref{Nonlinear1}) can be transformed into
\begin{linenomath*}
\begin{equation}
 \left\{
   \begin{aligned}
   \frac{dx_0}{dt} &= -(A_1-A_3)x_0^2-(A_2-A_3)x_0x_1+A_4x_0+A_5x_1+\gamma_2\\
   \frac{dx_1}{dt} &= -(A_2-A_3)x_1^2-(A_1-A_3)x_0x_1+A_6x_0+A_7x_1+\gamma_3\\
    \end{aligned}
   \right.
   \label{Nonlinear2}
\end{equation}
\end{linenomath*}
where
\begin{linenomath*}
\begin{equation}
A_4=\alpha_1-(\gamma_1-\gamma_4)-\gamma_2,
\end{equation}
\begin{equation}
A_5=\beta_2-\gamma_2,
\end{equation}
\begin{equation}
A_6=\alpha_2-\gamma_3,
\end{equation}
\begin{equation}
A_7=(\beta_1-\beta_2-\beta_3-\beta_4)-(\gamma_1-\gamma_4)-\gamma_3.
\end{equation}
\end{linenomath*}
When $\beta_2=\beta_3=\gamma_2=\gamma_3=0$, \emph{i.e.}
the phenotypic conversion rates are set to be zero,
we have
\begin{linenomath*}
\begin{equation}
 \left\{
   \begin{aligned}
   \frac{dx_0}{dt} &= -(A_1-A_3)x_0^2-(A_2-A_3)x_0x_1+A'_4x_0+A'_5x_1\\
   \frac{dx_1}{dt} &= -(A_2-A_3)x_1^2-(A_1-A_3)x_0x_1+A'_6x_0+A'_7x_1\\
    \end{aligned}
   \right.
   \label{Nonlinear3}
\end{equation}
\end{linenomath*}
where
\begin{linenomath*}
\begin{equation}
A'_4=\alpha_1-(\gamma_1-\gamma_4),
\end{equation}
\begin{equation}
A'_5=0,
\end{equation}
\begin{equation}
A'_6=\alpha_2,
\end{equation}
\begin{equation}
A'_7=(\beta_1-\beta_4)-(\gamma_1-\gamma_4).
\end{equation}
\end{linenomath*}
We term Eq. (\ref{Nonlinear2}) the \emph{reversible model},
and term Eq. (\ref{Nonlinear3}) the \emph{hierarchical model}.
The comparison of these two models is the main task in the following section.

In this section we try to give a comparative study of  Eqs. (\ref{Nonlinear2})
and (\ref{Nonlinear3}).

One trivial difference between the two models happens when
there is no CSC in the initial population. In this case,
since in the hierarchical model CSCs cannot be spontaneously created by other cell phenotypes,
there is no CSC in the population all the time,
whereas new-born CSCs can be generated from NSCCs
in the reversible model.
However, in reality it is quite impossible to completely purify
cancer cell lines without any CSCs by cell sorting. Based on this fact,
we only discuss the cases with positive initial states of CSCs afterwards.
Both long-term stable behavior and short-term transient dynamics
are investigated when comparing the two models.
Enlightened by \cite{gupta2011stochastic}, we are particularly concerned about
their predictions to the phenotypic
equilibrium and overshoot.

\subsection{Special case}

Before exploring the general models in Eqs. (\ref{Nonlinear2}) and (\ref{Nonlinear3}),
we first discuss a very specific case with $A_1=A_2=A_3$.
Investigating this special case can provide some valuable insights to the general analysis.

For the hierarchical model in this case,
\begin{linenomath*}
\begin{equation}
 \left\{
   \begin{aligned}
   \frac{dx_0}{dt} &= A'_4x_0\\
   \frac{dx_1}{dt} &= A'_6x_0\\
    \end{aligned}
   \right.
   \label{Linearhierarchicalmodel}
\end{equation}
\end{linenomath*}
It is easy to give their explicit solutions
\begin{linenomath*}
\begin{equation}
 \left\{
   \begin{array}{l}
  x_0(t) = x_0(0)e^{-(\alpha_2+\alpha_3)t}\\
   x_1(t) =\frac{\alpha_2}{\alpha_2+\alpha_3}x_0(0)(1-e^{-(\alpha_2+\alpha_3)t})+x_1(0)\\
    \end{array}
   \right.
\end{equation}
\end{linenomath*}
where $x_0(0)$ and $x_1(0)$ are the initial states. So we have
\begin{enumerate}
  \item Since $x_0(t)$ and $x_1(t)$ are both monotonic functions of $t$, they cannot perform
  overshoot.
  \item Note that
  \begin{linenomath*}
$$\lim_{t\rightarrow +\infty}x_0(t)=0,$$
$$\lim_{t\rightarrow +\infty}x_1(t)=\frac{\alpha_2}{\alpha_2+\alpha_3}x_0(0)+x_1(0),$$
\end{linenomath*}
CSCs proportion will eventually be zero, and the equilibrium proportions of $x_1$
depends on the initial states. Note that the phenotypic equilibrium corresponds to
the stabilization of the phenotypic \emph{mixture}
 that is \emph{independent} of the initial states
 \cite{gupta2011stochastic}, the hierarchical model fails to predict the phenotypic
 equilibrium.
\end{enumerate}

We turn our attention to the reversible model
\begin{linenomath*}
\begin{equation}
 \left\{
   \begin{aligned}
   \frac{dx_0}{dt} &= A_4x_0+A_5x_1+\gamma_2\\
   \frac{dx_1}{dt} &= A_6x_0+A_7x_1+\gamma_3\\
    \end{aligned}
   \right.
   \label{Linear3}
\end{equation}
\end{linenomath*}
For the phenotypic equilibrium, we have the following theorem:
\newtheorem{theorem}{Theorem}
\begin{theorem}
There exists unique stable fixed point $E=(x^*_0, x^*_1)$ in Eq. (\ref{Linear3}), where
$$x^*_0=\frac{A_5\gamma_3-A_7\gamma_2}{A_4A_7-A_5A_6}>0,\qquad x^*_1=\frac{A_6\gamma_2-A_4\gamma_3}{A_4A_7-A_5A_6}>0.$$
\label{Thm1}
\end{theorem}
The proof is put in \ref{appendixA}. Theorem \ref{Thm1} indicates that the phenotypic equilibrium
naturally holds in the special reversible model. We now consider whether this model
can perform overshoot. The main results are listed in Table \ref{table1} (see \ref{appendixB} for
more details).
In particular,
Fig. 2 illustrates the predictions by case 1 ($\Delta>0$):
\begin{itemize}
  \item \emph{Non-overshoot (left panel)}. Both $x_0$ and $x_1$ are monotonic functions.
  \item \emph{Asynchronous overshoot (middle panel)}. Either $x_0$ or $x_1$ can perform overshoot (The figure shows an example of
  $x_0$-overshoot), but they cannot perform overshoot simultaneously.
  \item \emph{Synchronous overshoot (right panel)}. $x_0$ and $x_1$ can perform overshoot simultaneously.
\end{itemize}
Case 2 ($\Delta=0$ and $A$ is diagonalizable) is the \emph{star} case (see \ref{appendixB}), overshoot can never happen.
For case 3 ($\Delta=0$ and $A$ is not diagonalizable), it belongs to the \emph{nodal} case (see \ref{appendixB}),
both asynchronous and synchronous overshoots can happen. Besides, in case 4 ($\Delta<0$), the model performs
damped oscillatory dynamics, corresponding to an interesting \emph{oscillating overshoot} (Fig. 3).

According to the above comparisons between Eqs. (\ref{Linearhierarchicalmodel}) and (\ref{Linear3}),
the reversible model shows richer dynamics than the hierarchical model in performing both the phenotypic
equilibrium and overshoot.

\begin{table}\footnotesize
\caption{The overshoot predicted by the special reversible model}
\begin{tabular}{|c|c|c|}
  \hline
  \hline
  & \textbf{Parameters} & \textbf{Behavior}\\
  \hline
  1 & $\Delta>0$ & no overshoot, asynchronous and synchronous overshoots\\
  \hline
2 & $\Delta=0$ and $A$ is diagonalizable  & no overshoot\\
  \hline
3 & $\Delta=0$ and $A$ is not diagonalizable  & asynchronous and synchronous overshoots\\
  \hline
4 & $\Delta<0$ & oscillating overshoot\\
  \hline
  \hline
\end{tabular}
\label{table1}
\end{table}

\begin{figure}
\begin{center}
\includegraphics[width=1\textwidth]{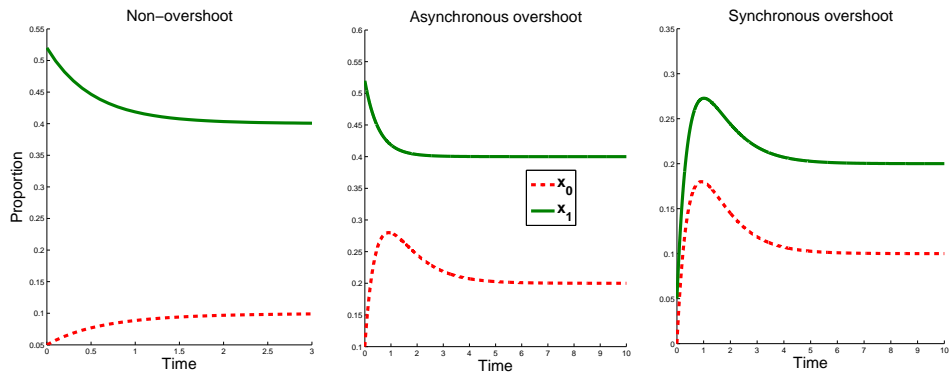}
\caption{Three types of predictions by the special reversible model in case 1 ($\Delta>0$).}
\end{center}
\end{figure}

\begin{figure}
\begin{center}
\includegraphics[width=1\textwidth]{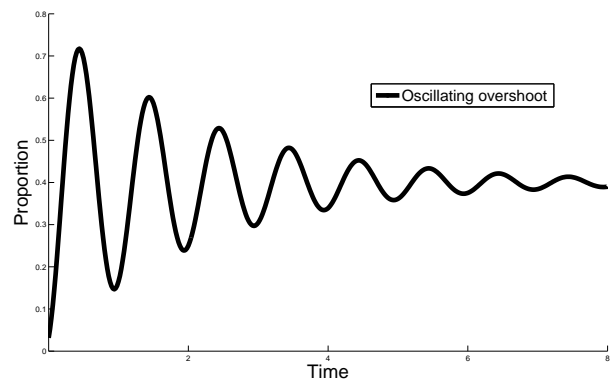}
\caption{Oscillating overshoot: Compared to the simple uphill-downhill overshoot in Fig 1,
oscillating overshoot shows multiple uphill-downhill movements with damped amplitude.}
\end{center}
\end{figure}

\subsection{General case}

We now remove the constraint of $A_1=A_2=A_3$ and further consider the general case in this section.
For the general hierarchical model Eq. (\ref{Nonlinear3}), the following theorem
states the condition under which the phenotypic mixture will finally stabilize:
\begin{theorem}
There exist three singular points in Eq. (\ref{Nonlinear3}): $E_0'=(0,0)$, $E_1'=(0,1)$ and $E_2'=(x_0^{*}, x_1^{*})$,
where
\begin{linenomath*}
$$x_1^{*}=\frac{A_4'A_6'}{(A_4'-A_7')(A_1'-A_3')+A_6'A_7'},\qquad x_0^{*}=\frac{(A_4'-A_7')x_1^{*}}{A_6'}.$$
Assume that $\alpha_1>\beta_1-\beta_4$ and $\alpha_1>\gamma_1-\gamma_4$, $E_2'=(x_0^{*}, x_1^{*})$ is the unique stable fixed point.
\end{linenomath*}
\label{Thm2}
\end{theorem}
The proof is put in \ref{appendixC}. Note that $\alpha_1$,
$\beta_1-\beta_4$ and $\gamma_1-\gamma_4$ are the diagonal
elements of the matrix $Q^*$ in Eq. (\ref{Matrix2}), representing
the per capita growth rates due to the corresponding cell phenotypes themselves.
Theorem \ref{Thm2} thus indicates that the general hierarchical model is capable of predicting
the phenotypic equilibrium only provided that the self-contributed growth rate by
CSCs is larger than that of more committed cancer cells. However,
this condition may not be satisfied in reality. It has been reported that
more committed cancer cells can have faster cycling time than stem-like cancer cells
in some cancers \cite{patrawala2005side,fillmore2008human,wang2014dynamics}.
In contrast, the following theorem shows the phenotypic equilibrium
of the general reversible model Eq. (\ref{Nonlinear1}):
\begin{theorem}
The reversible model Eq. (\ref{Nonlinear1})
has unique positive stable fixed point
\footnote{We should point out that, to complete the final proof,
technically it is necessary to add a small perturbation to the initial state
in few cases, see  \ref{appendixE}.}.
\label{Thm3}
\end{theorem}
The proof is put in \ref{appendixE}. Theorem \ref{Thm3} shows that
the result in the special case (Theorem \ref{Thm1}) can be extended to the general case.
By comparison of Theorems \ref{Thm2} and \ref{Thm3}, we can see that it is more likely
for the reversible model to correctly capture the phenotypic equilibrium.

In the remaining of this section, we discuss the overshoot performances of the two models by fitting them
to the cell-state dynamics in \cite{gupta2011stochastic}. We fitted Eqs. (\ref{Nonlinear3})
and (\ref{Nonlinear1}) to the SUM159 data shown in the Fig. 3
of \cite{gupta2011stochastic}. From Fig. 4,
we can see that the reversible model fits the data better than
the hierarchical model. In particular, note that both the overshoots of luminal cells and stem-like cells
were presented in the SUM159 data (see the left panel and right panel in the Fig. 3 of \cite{gupta2011stochastic}
respectively). Our result shows that, the reversible model can capture both of them, whereas the hierarchical
model can only fit the overshoot of luminal cells but it fails to capture the overshoot of stem-like cells.
In the left panel of Fig. 4, both the blue solid line (hierarchical model) and blue
discrete triangle-marker line (reversible model) fit the overshoot of luminal cells. Meanwhile, in
the right panel of Fig. 4, the red solid line (hierarchical model) shows a monotonic way up to the final equilibrium (non-overshoot),
but the red discrete triangle-marker line (reversible model) predict the overshoot of stem-like cells.

Actually, it has already been reported in previous literature that the overshoot of NSCCs proportion
can happen in classical CSC model \cite{weekes2014multicompartment}, but the
overshoot of CSCs proportion we think has added significance.
Note that in the hierarchical model the increase of CSCs relies only on their self-renewals.
When the initial fraction of CSCs is very limited, the transient increase
of CSCs proportion should be very slow (constraint by the limitation of cell division cycle),
so intuitively it is quite unlikely for the hierarchical model to capture the overshoot of
CSCs proportion. In contrast, the de-differentiation from NSCCs can effectively speed up
the accumulation of CSCs, it is more likely for the reversible model to capture the overshoot
of CSCs proportion. In other words, the overshoot of CSCs should be a better candidate than that of NSCCs to
characterize the phenotypic plasticity. Our numerical result just supports this idea.

\

\begin{figure}
\begin{center}
\includegraphics[width=1\textwidth]{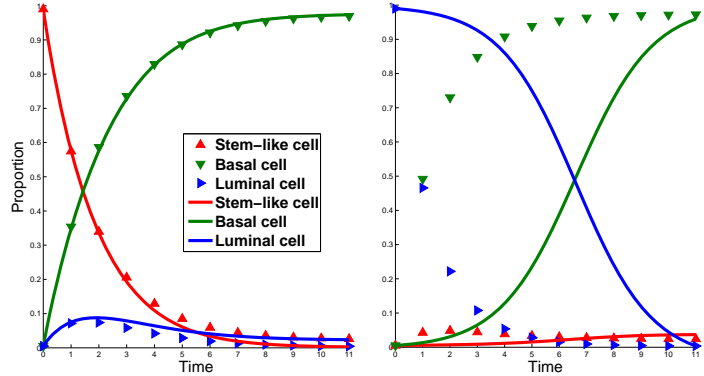}
\caption{Predictions of the general hierarchical model and reversible model. We used least square method to fit the two models
to the SUM159 data in \cite{gupta2011stochastic}. The left panel shows the purified stem-like cell case, where the initial states of stem-like cells, basal cells and luminal cells were assumed to be 99\%, 0.5\% and 0.5\% respectively. The right panel shows the purified luminal cell case, where the initial states of stem-like cells, basal cells and luminal cells were assumed to be 0.5\%, 0.5\% and 99\% respectively. In both panels, the solid lines are the predictions by the hierarchical model, and the discrete triangle-marker lines are the predictions by the reversible model.}
\end{center}
\end{figure}

\section{Conclusions}
\label{Conclusions}

In this paper, we have tried to characterize the phenotypic plasticity by giving a
comparative study between the reversible model and hierarchical model.
According to our results, the reversible model shows richer dynamics and better
predictions to experimental phenomena than the hierarchical model.

It should be noted that, to focus our attention to the reversible phenotypic
plasticity, the presented model does not include the biologically complex
mechanisms such as feedback controls and time delays.  It is undeniable that,
by including the feedback and delay mechanisms, the hierarchical model can
also show complex dynamics that are beyond the reach of the presented model
\cite{bernard2004bifurcations,stiehl2014impact,manesso2013dynamical,lander2009cell,liu2013nonlinear}.
For example, Bernard \emph{et al} showed that oscillations can occur in their hierarchical model
with feedbacks and delays \cite{bernard2004bifurcations}. Stiehl \emph{et al}
showed that the hierarchical model of hematopoiesis with nonlinear feedback can reproduce the
overshoot and damped oscillations observed in clinical data \cite{stiehl2014impact}.
We also think that the interplay between the feedback (or delay) mechanisms and the
phenotypic plasticity should be a very interesting issue.
As a staring point, it is suggested by our presented model that,
the phenotypic plasticity facilitates the heterogeneity of cancer in two ways:
it helps to sustain the coexistence of multiple phenotypes in cancer, and it also
accelerates the recovery process of cancer stem cells in the population.
These findings may shed some lights on the researches of the models
incorporating the biological complexities with the phenotypic plasticity.

Moreover, the presented model is deterministic, whereby stochastic effects are not
accounted for. Therefore, the branching model \cite{haccou2005branching}
concerning both the stochastic phenotypic conversions and proliferations of cancer cells
is another interesting research direction.
Finally, it is worth noting that, in addition to the long-term stabilization,
the short-term transient dynamics also deserve special attention
in the study of biological population dynamics \cite{hastings2004transients}.
In this work we have seen the significance of the overshoot in characterizing
the phenotypic plasticity. It is interesting to explore more types of transient dynamics
and their biological significance in future researches.

\section*{Acknowledgements}
X. C. and X. Z. acknowledge the support by NSFC (No. 11371161).
M. Y. acknowledges NSFC (Nos. 11275259 and 91330113).
D. Z. acknowledges the generous sponsorship from the National Natural
Sciences Foundation of China (No. 11401499), the Natural Science Foundation
of Fujian Province of China (No. 2015J05016),and the
Fundamental Research Funds for the Central Universities in China (No. 20720140524).

\appendix

\section{Proof of Theorem \ref{Thm1}}
\label{appendixA}

\newtheorem*{theorem*}{Theorem}
\begin{theorem*}[page \pageref{Thm1}, Theorem \ref{Thm1}]
There exists unique stable fixed point $E=(x^*_0, x^*_1)$ in Eq. (\ref{Linear3}), where
$$x^*_0=\frac{A_5\gamma_3-A_7\gamma_2}{A_4A_7-A_5A_6}>0,\qquad x^*_1=\frac{A_6\gamma_2-A_4\gamma_3}{A_4A_7-A_5A_6}>0.$$
\end{theorem*}
\begin{proof}
Let
\begin{linenomath*}
\begin{equation}
 \left\{
   \begin{array}{l}
   A_4x_0+A_5x_1+\gamma_2=0\\
    A_6x_0+A_7x_1+\gamma_3=0,\\
    \end{array}
   \right.
\end{equation}
\end{linenomath*}
where $A_4=-(\alpha_2+\alpha_3+\gamma_2)<0$, $A_7=-(\beta_2+\beta_3+\gamma_3)<0$,
$A_5=\beta_2-\gamma_2$, $A_6=\alpha_2-\gamma_3$.
Note that
\begin{linenomath*}
$$ A_4A_7-A_5A_6>A_5\gamma_3-A_7\gamma_2>0,\qquad A_4A_7-A_5A_6>A_6\gamma_2-A_4\gamma_3>0,$$
\end{linenomath*}
it is easy to see that $E=(x^*_0, x^*_1)$ is the only fixed point.
To proof the stability of $E$, let
$x_0=w_0+x^*_0$,~$x_1=w_1+x^*_1$, then Eq. (\ref{Linear3}) is transformed to
\begin{linenomath*}
\begin{equation}
 \left\{
   \begin{aligned}
  \frac{dw_0}{dt} &= A_4w_0+A_5w_1\\
  \frac{dw_1}{dt} &= A_6w_0+A_7w_1\\
    \end{aligned}
   \right.
   \label{Linear4}
\end{equation}
\end{linenomath*}
Let
$\vec{w}=(w_0,w_1)^T$ and
\begin{linenomath*}
\begin{equation}
A=\left(\begin{array}{cc}
A_4 & A_5   \\
A_6 & A_7  \\
\end{array}\right),
\label{Matrix3}
\end{equation}
\end{linenomath*}
Eq. (\ref{Linear4}) can be expressed as
\begin{linenomath*}
$$\frac{d\vec{w}}{dt} = A\vec{w}.$$
\end{linenomath*}
$E'=(0,0)$ becomes the only fixed point and its characteristic equation
is
\begin{linenomath*}
\begin{equation}
\lambda^2-(A_4+A_7)\lambda+A_4A_7-A_5A_6=0.
\label{char}
\end{equation}
\end{linenomath*}
Since
\begin{linenomath*}
$$\lambda_1+\lambda_2=A_4+A_7<0,\qquad \lambda_1\lambda_2=A_4A_7-A_5A_6>0,$$
\end{linenomath*}
the real parts of $\lambda_1$ and $\lambda_2$ are both negative.
Hence $E'$ is stable in Eq. (\ref{Linear4}),
and then $E$ is the unique stable fixed point of Eq. (\ref{Linear3}).
The proof is completed.
\end{proof}

\section{The overshoot predicted by the special reversible model}
\label{appendixB}

Let
\begin{equation}
\Delta=(A_4-A_7)^2+4A_5A_6
\label{delta}
\end{equation}
be the discriminant of (\ref{char}).
Since it has been shown that the real parts of
$\lambda_1$ and $\lambda_2$ are negative,
the special reversible model Eq. (\ref{Linear3}) can be classified into the following four
cases (For the terminologies used below, please refer to Chapter 5 in \cite{birkhoff1969ordinary}):

\begin{figure}[!htb]
\begin{center}
\includegraphics[width=0.4\textwidth]{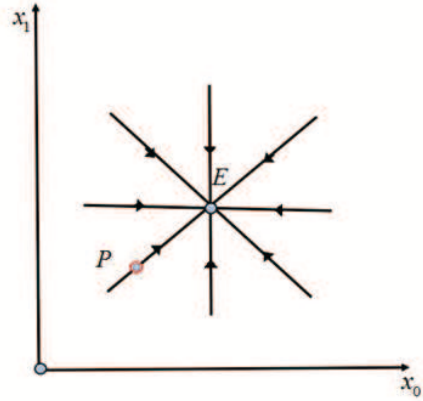}
\caption{The phase diagrammatic sketch of the case that $\Delta=0$ and $A$ is diagonalizable.}
\end{center}
\end{figure}

\begin{figure}[!htb]
\begin{center}
\subfigure[]{\includegraphics[width=0.4\textwidth]{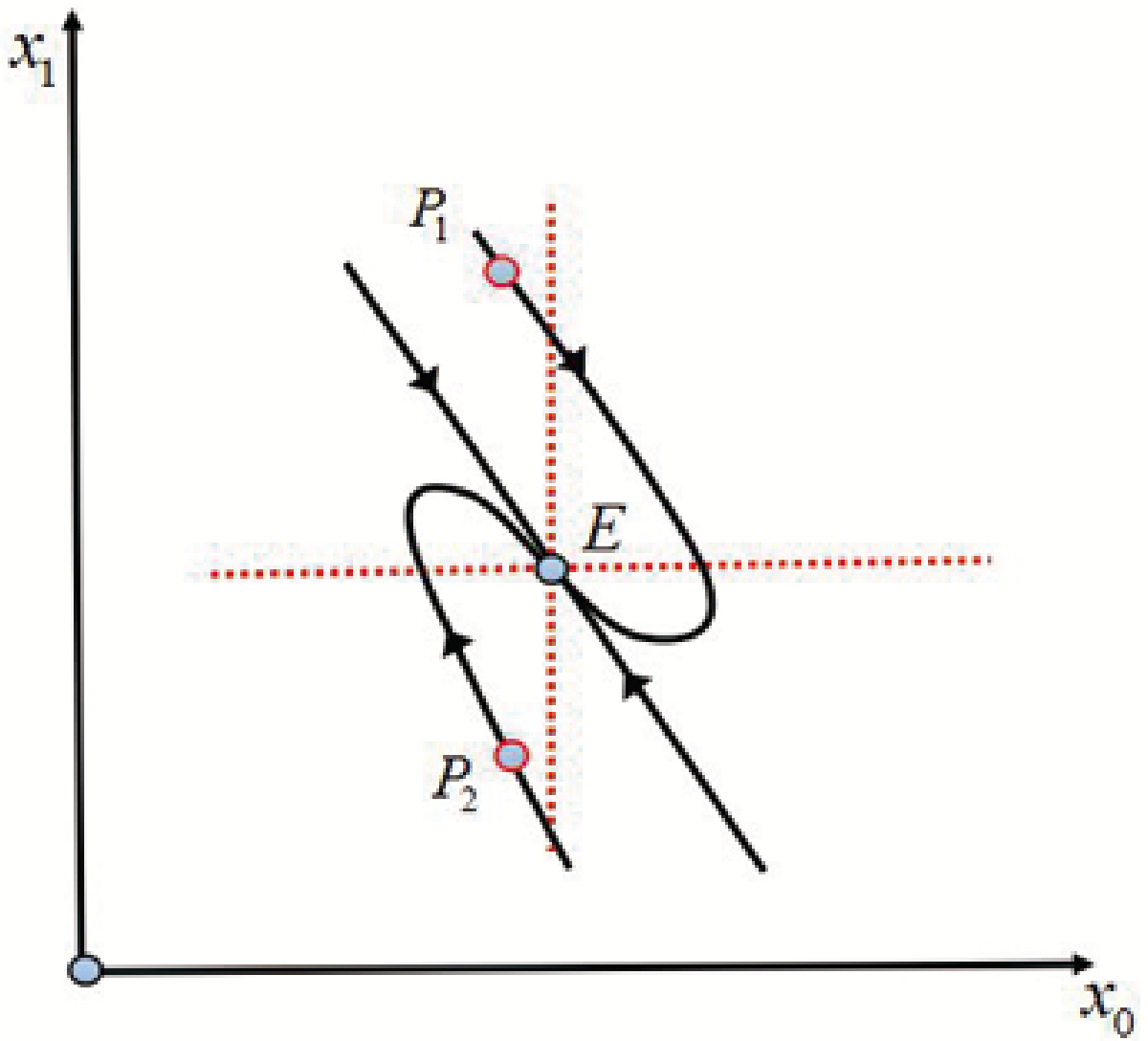}}
\subfigure[]{\includegraphics[width=0.4\textwidth]{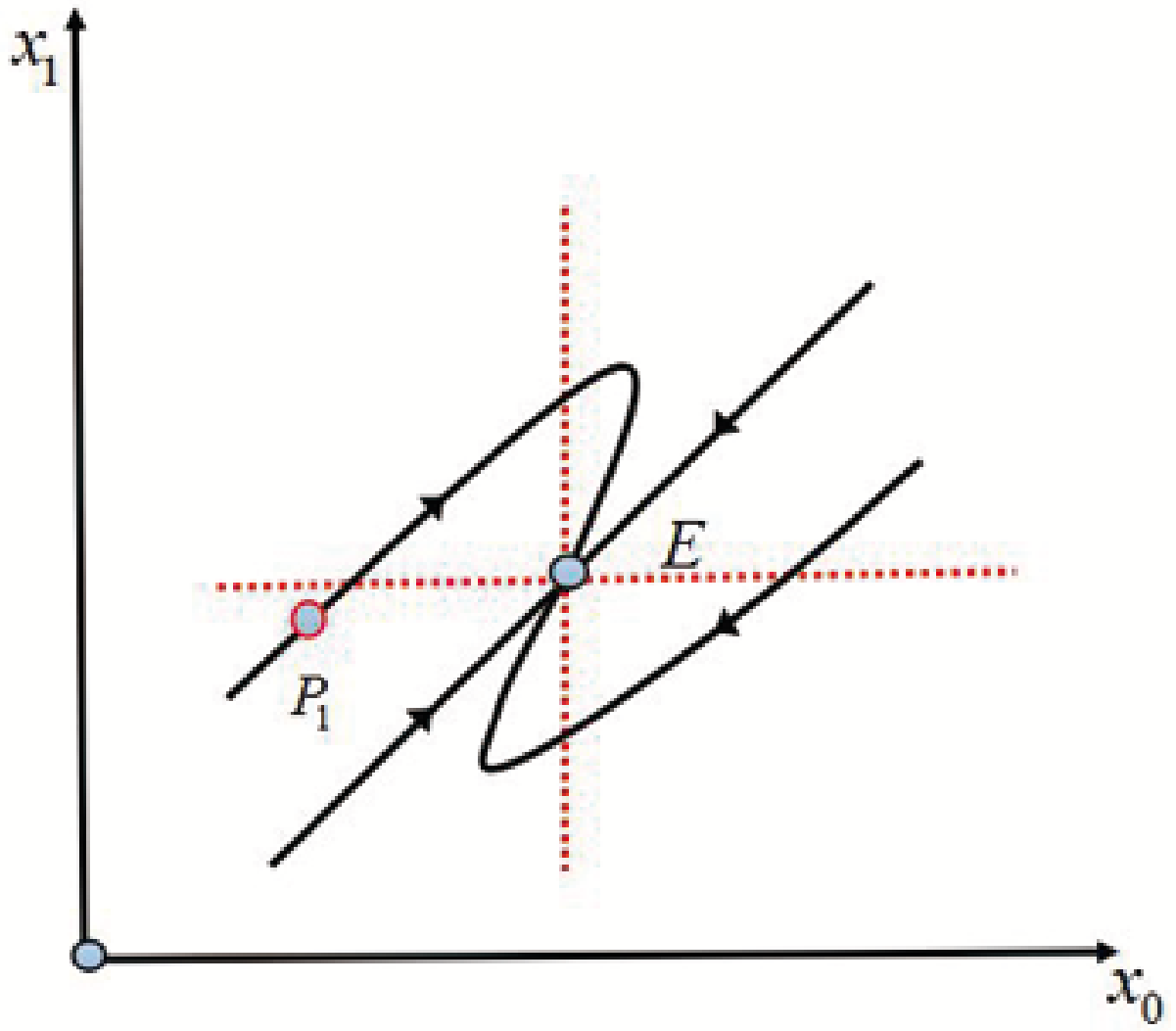}}
\caption{The phase diagrammatic sketch of the case that $\Delta=0$ and $A$ is not diagonalizable.
(a) Asynchronous overshoot: When starting from $P_1$,
$x_0$ can first exceed $x^*_0$ and then decrease to it (overshoot),
whereas $x_1$ cannot perform overshoot. However, when starting from $P_2$,
$x_1$ can perform overshoot, but $x_0$ cannot perform overshoot. That is, either $x_0$
or $x_1$ can perform overshoot, but they cannot perform overshoot simultaneously.
(b) Synchronous overshoot: When starting from $P_1$, both $x_0$ and $x_1$
can exceed $E$ and then decrease to it. That is, they can perform overshoot simultaneously.}
\end{center}
\end{figure}

\begin{enumerate}
  \item $\Delta>0$. In this case the matrix $A$ in Eq. (\ref{Matrix3})
  has two unequal negative eigenvalues. Thus $E'$
  is the \emph{nodal point} of Eq. (\ref{Linear4}), correspondingly $E$ is the \emph{nodal point}
  of Eq. (\ref{Linear3}). The solution can be expressed as
  \begin{linenomath*}
\begin{equation}
\vec{x}=P(e^{\lambda_1t}k_1,e^{\lambda_2t}k_2)^T,
\end{equation}
\label{nodal}
\end{linenomath*}
where
\begin{equation}
P=\left(\begin{array}{cc}
p_{11} & p_{12}   \\
p_{21} & p_{22}  \\
\end{array}\right)
\end{equation}
is a matrix determined by the coefficients of Eq. (\ref{Linear4}),
both $k_1$ and $k_2$ are determined by the initial states. Different types of transient dynamics can
happen for different parameter ranges. By numerical simulations, it is shown in Fig. 2 that
the model can perform non-overshoot, asynchronous overshoot and synchronous overshoot respectively.

  \item $\Delta=0$ and $A$ is diagonalizable. Since
  $E$ is the \emph{star point} in this case (Fig. B.5). There is no overshoot
  of $x_0$ and $x_1$.
  \item $\Delta=0$ and $A$ is not diagonalizable. In this case $E$
  is also called \emph{nodal point}. From Fig. B.6, we can see that
  the model can perform both asynchronous and synchronous overshoots.
  If $(A_4-A_7)/A_5>0$, the model can perform asynchronous overshoot (Fig. B.6(a)).
  If $(A_4-A_7)/A_5<0$, the model can perform synchronous overshoot (Fig. B.6(b)).
  \item $\Delta<0$. $A$ has two conjugate complex eigenvalues with
  negative real part. $E$ is the \emph{focal point}, resulting in
 oscillating overshoot (Fig. 3).
\end{enumerate}

\section{Proof of Theorem \ref{Thm2}}
\label{appendixC}

\begin{theorem*}[page \pageref{Thm2}, Theorem \ref{Thm2}]
There exist three singular points in Eq. (\ref{Nonlinear3}): $E_0'=(0,0)$, $E_1'=(0,1)$ and $E_2'=(x_0^{*}, x_1^{*})$,
where
\begin{linenomath*}
$$x_1^{*}=\frac{A_4'A_6'}{(A_4'-A_7')(A_1'-A_3')+A_6'A_7'},\qquad x_0^{*}=\frac{(A_4'-A_7')x_1^{*}}{A_6'}.$$
Assume that $\alpha_1>\beta_1-\beta_4$ and $\alpha_1>\gamma_1-\gamma_4$, $E_2'=(x_0^{*}, x_1^{*})$ is the unique stable fixed point.
\end{linenomath*}
\end{theorem*}
\begin{proof}
To obtain the fixed points of Eq. (\ref{Nonlinear3}), we consider
\begin{linenomath*}
\begin{equation}
 \left\{
   \begin{array}{l}
   -(A_1'-A_3')x_0^2-(A_2'-A_3')x_0x_1+A_4'x_0=0\\
    -(A_2'-A_3')x_1^2-(A_1'-A_3')x_0x_1+A_6'x_0+A_7'x_1=0\\
    \end{array}
   \right.
\end{equation}
\end{linenomath*}
For the first equation
$-(A_1'-A_3')x_0^2-(A_2'-A_3')x_0x_1+A_4'x_0=0$,
we have either $x_0=0$ or
$-(A_1'-A_3')x_0-(A_2'-A_3')x_1+A_4'=0.$\\
\emph{(i)} When $x_0=0$, based on the second equation we have $x_1=0$ or
$x_1=A_7'/(A_2'-A_3')=1,$
that is, we have two fixed points
\begin{linenomath*}
         $$E_0'=(0,0),\qquad E_1'=(0,1).$$
\end{linenomath*}
\emph{(ii)} When $-(A_1'-A_3')x_0-(A_2'-A_3')x_1+A_4'=0,$
the fixed point is
\begin{linenomath*}
$$E_2'=(x_0^{*},x_1^{*}),$$
\end{linenomath*}
where
\begin{linenomath*}
$$x_1^{*}=\frac{A_4'A_6'}{(A_4'-A_7')(A_1'-A_3')+A_6'A_7'},\qquad x_0^{*}=\frac{(A_4'-A_7')x_1^{*}}{A_6'}.$$
\end{linenomath*}
Since
$\alpha_1>\beta_1-\beta_4$ and $\alpha_1>\gamma_1-\gamma_4$,
we have $0<x_0^{*}<1$ and $0<x_1^{*}<1.$
Let\begin{linenomath*}
$$f(x_0,x_1)=-(A_1'-A_3')x_0^2-(A_2'-A_3')x_0x_1+A_4'x_0,$$
$$g(x_0,x_1)=-(A_2'-A_3')x_1^2-(A_1'-A_3')x_0x_1+A_6'x_0+A_7'x_1,$$
then
$$f_{x_0}'=-2(A_1'-A_3')x_0-(A_2'-A_3')x_1+A_4',~~f_{x_1}'=-(A_2'-A_3')x_0,$$
$$g_{x_0}'=-(A_1'-A_3')x_1+A_6',~~g_{x_1}'=-2(A_2'-A_3')x_1-(A_1'-A_3')x_0+A_7'.$$
\end{linenomath*}
The Jacobian matrix of Eq. (\ref{Nonlinear3}) is given by
\begin{linenomath*}
$$J=\begin{pmatrix} f_{x_0}' & f_{x_1}' \\
g_{x_0}'& g_{x_1}' \\
\end{pmatrix}.$$
\end{linenomath*}
Let us first consider the stability of $E_0'=(0,0)$,
\begin{linenomath*}
$$J(E_0')=\begin{pmatrix} A_4' & 0 \\
A_6'& A_7' \\
\end{pmatrix},$$
\end{linenomath*}
its characteristic equation is
\begin{linenomath*}
$$(\lambda-A_4')(\lambda-A_7')=0.$$
Note that
$$\alpha_1>\beta_1-\beta_4,~~~\alpha_1>\gamma_1-\gamma_4,$$
\end{linenomath*}
then $\lambda_1=A_4'>0$, implying that $E_0'=(0,0)$ is unstable.

\begin{linenomath*}
Similarly, for $E_1'=(0,1)$
$$J(E_1')=\begin{pmatrix} A_4'-A_7' & 0 \\
A_6'-A_1'+A_3'& -A_7' \\
\end{pmatrix},$$
$$(\lambda-(A_4'-A_7'))(\lambda+A_7')=0,$$
\end{linenomath*}
$\lambda_1=A_4'-A_7'>0$, $E_1'=(0,1)$ is also unstable.

For $E_2'=(x_0^{*}, x_1^{*})$,
\begin{linenomath*}
$$J(E_2')=\begin{pmatrix} f_{x_0} & f_{x_1} \\
g_{x_0}& g_{x_1} \\
\end{pmatrix},$$
\end{linenomath*}
where
\begin{linenomath*}
$$f_{x_0}=f_{x_0}'(x_0^{*},x_1^{*}),\qquad f_{x_1}=f_{x_1}'(x_0^{*},x_1^{*}),\qquad g_{x_0}=g_{x_0}'(x_0^{*},x_1^{*}),\qquad g_{x_1}=g_{x_1}'(x_0^{*},x_1^{*}).$$
\end{linenomath*}
Then we have
\begin{linenomath*}
$$\lambda^2-(f_{x_0}+g_{x_1})\lambda+f_{x_0}g_{x_1}-f_{x_1}g_{x_0}=0,$$
$$\triangle=[-(\lambda_1+\lambda_2)]^{2}-4\lambda_1\lambda_2=A_7'^{2}>0,$$
$$\lambda_1+\lambda_2=A_7'-2A_4'<0,$$
$$\lambda_1\lambda_2=A_4'(A_4'-A_7')>0.$$
\end{linenomath*}
Then
\begin{linenomath*}
$$\lambda_1<0,\qquad \lambda_2<0.$$
\end{linenomath*}
It is shown that $E_2'$ is a stable nodal point, and then it is the unique stable fixed point in Eq. (\ref{Nonlinear3}).
The proof is completed.
\end{proof}

\section{Proof of Theorem \ref{Thm3}}
\label{appendixE}

\begin{theorem*}[page \pageref{Thm3}, Theorem \ref{Thm3}]
The reversible model Eq. (\ref{Nonlinear1})
has unique positive stable fixed point.
\end{theorem*}

\begin{proof}
Our proof is valid for general multi-phenotypic cases, so in this section the dimension of the matrix
$Q$ in Eq (\ref{Matrix1}) is not restricted to three any more, but $n$ in general.
Without loss of generality, we only prove the theorem when
$Q$ is positive (all the elements of $Q$ are positive)
\footnote{Note that the off-diagonal elements of $Q$ are positive,
$Q+\tau I$ ($I$ is the identity matrix) can be positive provided $\tau$ is large enough. It is easy to show that when $\lambda$ is an eigenvalue of $Q+\tau I$,
$\lambda-\tau$ is an eigenvalue of $Q$ accordingly. That is, the eigenvalue structure of $Q$ is the same as that of $Q+\tau I$. So we just need to
consider $Q+\tau I$ instead of $Q$.}.
According to the well-known Perron-Frobenius theory (see chapter 1 in \cite{seneta1981non}), $Q$ has a real eigenvalue $\lambda_1$ satisfying
$Re\lambda<\lambda_1$ for any other eigenvalue $\lambda$ of $Q$
(called the Perron-Frobenius eigenvalue) and $\lambda_1$ is
simple (a simple root of the characteristic equation of $Q$).

Since $\lambda_1$ is simple, the solution of Eq. (\ref{Linear1}) can be expressed as
\begin{linenomath*}
\begin{equation}
\begin{split}
&\overrightarrow{X_t}=c_{1,1}\vec{u}e^{\lambda_{1}t}+\sum_{j=2}^{m}\sum_{l=1}^{m_j}c_{j,l}\sum_{i=1}^{m_j}\vec{r^{j}_{l,i}}t^{i-1}e^{\lambda_{j}t},
\end{split}
\label{solution}
\end{equation}
\end{linenomath*}
where $\lambda_{1},\lambda_{2},\cdots{}\lambda_{m}$ are the different eigenvalues of $Q$,
$m_j$ is the algebraic multiplicity of $\lambda_{j}$,
$\vec{u}$ is the normalized ($u_1+u_2+...+u_n=1$) right eigenvector of $\lambda_1$,
$\vec{r^j_{l,i}}$ is the corresponding eigenvector of $\lambda_{j}$,
$c_{j,l}$ is determined by initial states. Suppose $c_{1,1}\neq{}0$, since Re$\lambda_i<\lambda_1~(i\neq 1)$,
\begin{linenomath*}
\begin{equation*}
\begin{split}
&\frac{\overrightarrow{X_t}}{c_{1,1}e^{\lambda_{1}t}}=\vec{u}+\sum_{j=2}^{n}\sum_{l=1}^{n_j}
\frac{c_{j,l}}{c_{1,1}}\sum_{i=1}^{n_j}\vec{r^{j}_{l,i}}t^{i-1}e^{(\lambda_{j}-\lambda_1)t}\rightarrow \vec{u}.
\end{split}
\end{equation*}
\end{linenomath*}
Thus
\begin{linenomath*}
$$\overrightarrow{x_t}=\frac{\overrightarrow{X_t}}{W}
=\frac{\overrightarrow{X_t}/c_{1,1}e^{\lambda_{1}t}}{(X_1+...+X_n)/c_{1,1}e^{\lambda_{1}t}}
\rightarrow \frac{\vec{u}}{u_1+u_2+...+u_n}=\vec{u}.$$
\end{linenomath*}
Noticing that Perron-Frobenius theory ensures that all the components of $\vec{u}$ are positive,
\emph{i.e.} $u>0$.

Before we complete the proof, we need to discuss $c_{1,1}=0$. In this case, the above method
does not work. However, since fluctuations are inevitable in real world,
$c_{1,1}=0$ will hardly happen in reality. To show this, let $t=0$ in Eq. (\ref{solution})
\begin{linenomath*}
\begin{equation*}
\begin{split}
&c_{1,1}\vec{u}+\sum_{j=2}^{n}\sum_{l=1}^{n_j}c_{j,l}\vec{r^{j}_{l,1}}=\langle\overrightarrow{X_0}\rangle^T
\end{split}
\end{equation*}
\end{linenomath*}
This is a linear equation of $c_{j,l}$. For $c_{1,1}$ we have
\begin{linenomath*}
\begin{equation*}
\begin{split}
&c_{1,1}=\frac{\textrm{det}|B^*|}{\textrm{det}|B|},
\end{split}
\end{equation*}
\end{linenomath*}
where $B=[\vec{u}\ \vec{r^{2}_{1,1}}\ \vec{r^{2}_{2,1}}\cdots \vec{r^{n}_{n_n,1}}]$, $B^*$ is just $B$ with
its first column replaced by $\langle\overrightarrow{X_0}\rangle^T$.
It is easy to add a small perturbation $\varepsilon{}\vec{v}$ to $\langle\overrightarrow{X_0}\rangle^T$,
so that all the columns of $B^*$ are linear independent, hence $c_{1,1}\neq{}0$.

\end{proof}


\begin{thebibliography}{10}

\bibitem{reya2001stem}
T.~Reya, S.~Morrison, M.~Clarke, I.~Weissman, Stem cells, cancer, and cancer
  stem cells, Nature 414 (2001) 105--111.

\bibitem{jordan2006cancer}
C.~Jordan, M.~Guzman, M.~Noble, Cancer stem cells, N. Engl. J. Med. 355~(12)
  (2006) 1253--1261.

\bibitem{dalerba2007cancer}
P.~Dalerba, R.~Cho, M.~Clarke, Cancer stem cells: models and concepts, Annu.
  Rev. Med. 58 (2007) 267--284.

\bibitem{meyer2009dynamic}
M.~Meyer, J.~Fleming, M.~Ali, M.~Pesesky, E.~Ginsburg, B.~Vonderhaar, Dynamic
  regulation of cd24 and the invasive, cd44poscd24neg phenotype in breast
  cancer cell lines, Breast Cancer Res. 11~(6) (2009) R82.

\bibitem{chaffer2013poised}
C.~L. Chaffer, N.~D. Marjanovic, T.~Lee, G.~Bell, C.~G. Kleer, F.~Reinhardt,
  A.~C. D¡¯Alessio, R.~A. Young, R.~A. Weinberg, Poised chromatin at the zeb1
  promoter enables breast cancer cell plasticity and enhances tumorigenicity,
  Cell 154~(1) (2013) 61--74.

\bibitem{quintana2010phenotypic}
E.~Quintana, M.~Shackleton, H.~R. Foster, D.~R. Fullen, M.~S. Sabel, T.~M.
  Johnson, S.~J. Morrison, Phenotypic heterogeneity among tumorigenic melanoma
  cells from patients that is reversible and not hierarchically organized,
  Cancer Cell 18~(5) (2010) 510--523.

\bibitem{yang2012dynamic}
G.~Yang, Y.~Quan, W.~Wang, Q.~Fu, J.~Wu, T.~Mei, J.~Li, Y.~Tang, C.~Luo,
  Q.~Ouyang, et~al., Dynamic equilibrium between cancer stem cells and non-stem
  cancer cells in human sw620 and mcf-7 cancer cell populations, Br. J. Cancer
  106~(9) (2012) 1512--1519.

\bibitem{fessler2015endothelial}
E.~Fessler, T.~Borovski, J.~P. Medema, Endothelial cells induce cancer stem
  cell features in differentiated glioblastoma cells via bfgf, Mol. cancer
  14~(1) (2015) 157.

\bibitem{gupta2011stochastic}
P.~Gupta, C.~Fillmore, G.~Jiang, S.~Shapira, K.~Tao, C.~Kuperwasser, E.~Lander,
  Stochastic state transitions give rise to phenotypic equilibrium in
  populations of cancer cells, Cell 146~(4) (2011) 633--644.

\bibitem{kussell2005phenotypic}
E.~Kussell, S.~Leibler, Phenotypic diversity, population growth, and
  information in fluctuating environments, Science 309~(5743) (2005)
  2075--2078.

\bibitem{dos2013possible}
R.~V. dos Santos, L.~M. da~Silva, A possible explanation for the variable
  frequencies of cancer stem cells in tumors, PLoS One 8~(8) (2013) e69131.

\bibitem{dos2013noise}
R.~V. dos Santos, L.~M. da~Silva, The noise and the kiss in the cancer stem
  cells niche, J. Theor. Biol. 335~(21) (2013) 79--87.

\bibitem{leder2014mathematical}
K.~Leder, K.~Pitter, Q.~LaPlant, D.~Hambardzumyan, B.~D. Ross, T.~A. Chan,
  E.~C. Holland, F.~Michor, Mathematical modeling of pdgf-driven glioblastoma
  reveals optimized radiation dosing schedules, Cell 156~(3) (2014) 603--616.

\bibitem{wang2014dynamics}
W.~Wang, Y.~Quan, Q.~Fu, Y.~Liu, Y.~Liang, J.~Wu, G.~Yang, C.~Luo, Q.~Ouyang,
  Y.~Wang, Dynamics between cancer cell subpopulations reveals a model
  coordinating with both hierarchical and stochastic concepts, PLoS One 9~(1)
  (2014) e84654.

\bibitem{zhoup2013opulation}
D.~Zhou, D.~Wu, Z.~Li, M.~Qian, M.~Q. Zhang, Population dynamics of cancer
  cells with cell state conversions, Quant. Biol. 1~(3) (2013) 201--208.

\bibitem{zhou2014multi}
D.~Zhou, Y.~Wang, B.~Wu, A multi-phenotypic cancer model with cell plasticity,
  J. Theor. Biol. 357 (2014) 35--45.

\bibitem{chen2014mathematical}
C.~Chen, W.~T. Baumann, J.~Xing, L.~Xu, R.~Clarke, J.~J. Tyson, Mathematical
  models of the transitions between endocrine therapy responsive and resistant
  states in breast cancer, J. R. Soc. Interface. 11~(96) (2014) 20140206.

\bibitem{zhou2014nonequilibrium}
J.~X. Zhou, A.~O. Pisco, H.~Qian, S.~Huang, Nonequilibrium population dynamics
  of phenotype conversion of cancer cells, PLoS One 9~(12) (2014) e110714.

\bibitem{jilkine2014effect}
A.~Jilkine, R.~N. Gutenkunst, E.~Wang, Effect of dedifferentiation on time to
  mutation acquisition in stem cell-driven cancers, PLoS Comput. Biol. 10~(3)
  (2014) e1003481.

\bibitem{jia2013overshoot}
C.~Jia, M.~Qian, D.~Jiang, Overshoot in biological systems modeled by markov
  chains: a nonequilibrium dynamic phenomenon, IET Syst. Biol. 8~(4) (2014)
  138--145.

\bibitem{zapperi2012cancer}
S.~Zapperi, C.~La~Porta, Do cancer cells undergo phenotypic switching? the case
  for imperfect cancer stem cell markers, Sci. Rep. 2 (2012) 441.

\bibitem{easwaran2014cancer}
H.~Easwaran, H.-C. Tsai, S.~B. Baylin, Cancer epigenetics: Tumor heterogeneity,
  plasticity of stem-like states, and drug resistance, Mol. Cell 54~(5) (2014)
  716--727.

\bibitem{chang2008transcriptome}
H.~H. Chang, M.~Hemberg, M.~Barahona, D.~E. Ingber, S.~Huang,
  Transcriptome-wide noise controls lineage choice in mammalian progenitor
  cells, Nature 453~(7194) (2008) 544--547.

\bibitem{chaffer2011normal}
C.~Chaffer, I.~Brueckmann, C.~Scheel, A.~Kaestli, P.~Wiggins, L.~Rodrigues,
  M.~Brooks, F.~Reinhardt, Y.~Su, K.~Polyak, et~al., Normal and neoplastic
  nonstem cells can spontaneously convert to a stem-like state, Proc. Natl.
  Acad. Sci. USA 108~(19) (2011) 7950--7955.

\bibitem{patrawala2005side}
L.~Patrawala, T.~Calhoun, R.~Schneider-Broussard, J.~Zhou, K.~Claypool, D.~G.
  Tang, Side population is enriched in tumorigenic, stem-like cancer cells,
  whereas abcg2+ and abcg2- cancer cells are similarly tumorigenic, Cancer Res.
  65~(14) (2005) 6207--6219.

\bibitem{fillmore2008human}
C.~M. Fillmore, C.~Kuperwasser, Human breast cancer cell lines contain
  stem-like cells that self-renew, give rise to phenotypically diverse progeny
  and survive chemotherapy, Breast Cancer Res. 10~(2) (2008) R25.

\bibitem{johnston2007mathematical}
M.~D. Johnston, C.~M. Edwards, W.~F. Bodmer, P.~K. Maini, S.~J. Chapman,
  Mathematical modeling of cell population dynamics in the colonic crypt and in
  colorectal cancer, Proc. Natl. Acad. Sci. USA 104 (2007) 4008--4013.

\bibitem{liu2013nonlinear}
X.~Liu, S.~Johnson, S.~Liu, D.~Kanojia, W.~Yue, U.~Singh, Q.~Wang, Q.~Wang,
  Q.~Nie, H.~Chen, Nonlinear growth kinetics of breast cancer stem cells:
  implications for cancer stem cell targeted therapy., Sci. Rep. 3 (2013) 2473.

\bibitem{pei2015fluctuation}
Q.-m. Pei, X.~Zhan, L.-j. Yang, J.~Shen, L.-f. Wang, K.~Qui, T.~Liu,
  J.~Kirunda, A.~Yousif, A.-b. Li, et~al., Fluctuation and noise propagation in
  phenotypic transition cascades of clonal populations, Phys. Rev. E 92~(1)
  (2015) 012721.

\bibitem{boman2007symmetric}
B.~Boman, M.~Wicha, J.~Fields, O.~Runquist, Symmetric division of cancer stem
  cells--a key mechanism in tumor growth that should be targeted in future
  therapeutic approaches, Clin. Pharmacol. Ther. 81~(6) (2007) 893--898.

\bibitem{dingli2007symmetric}
D.~Dingli, A.~Traulsen, F.~Michor, (a) symmetric stem cell replication and
  cancer, PLoS Comput. Biol. 3~(3) (2007) e53.

\bibitem{marciniak2009modeling}
A.~Marciniak-Czochra, T.~Stiehl, A.~D. Ho, W.~J{\"a}ger, W.~Wagner, Modeling of
  asymmetric cell division in hematopoietic stem cells-regulation of
  self-renewal is essential for efficient repopulation, Stem Cells Dev. 18~(3)
  (2009) 377--386.

\bibitem{morrison2006asymmetric}
S.~J. Morrison, J.~Kimble, Asymmetric and symmetric stem-cell divisions in
  development and cancer, Nature 441~(7097) (2006) 1068--1074.

\bibitem{lander2009cell}
A.~D. Lander, K.~K. Gokoffski, F.~Y. Wan, Q.~Nie, A.~L. Calof, Cell lineages
  and the logic of proliferative control, PLoS Biol. 7~(1) (2009) 84.

\bibitem{lo2009feedback}
W.-C. Lo, C.-S. Chou, K.~K. Gokoffski, F.~Y.-M. Wan, A.~D. Lander, A.~L. Calof,
  Q.~Nie, Feedback regulation in multistage cell lineages, Math. Biosci. Eng.
  6~(1) (2009) 59.

\bibitem{komarova2013principles}
N.~Komarova, Principles of regulation of self-renewing cell lineages., PLoS One
  8~(9) (2013) e72847.

\bibitem{weekes2014multicompartment}
S.~L. Weekes, B.~Barker, S.~Bober, K.~Cisneros, J.~Cline, A.~Thompson,
  L.~Hlatky, P.~Hahnfeldt, H.~Enderling, A multicompartment mathematical model
  of cancer stem cell-driven tumor growth dynamics, Bull. Math. Biol. 76~(7)
  (2014) 1762--1782.

\bibitem{bernard2004bifurcations}
S.~Bernard, J.~B{\'e}lair, M.~C. Mackey, Bifurcations in a white-blood-cell
  production model, C. R. Biologies 327~(3) (2004) 201--210.

\bibitem{stiehl2014impact}
T.~Stiehl, A.~Ho, A.~Marciniak-Czochra, The impact of cd34\&plus; cell dose on
  engraftment after scts: personalized estimates based on mathematical
  modeling, Bone Marrow Transpl. 49~(1) (2014) 30--37.

\bibitem{manesso2013dynamical}
E.~Manesso, J.~Teles, D.~Bryder, C.~Peterson, Dynamical modelling of
  haematopoiesis: an integrated view over the system in homeostasis and under
  perturbation, J. R. Soc. Interface 10~(80) (2013) 20120817.

\bibitem{haccou2005branching}
P.~Haccou, P.~Jagers, V.~A. Vatutin, Branching processes: variation, growth,
  and extinction of populations, Cambridge University Press, Cambridge, 2005.

\bibitem{hastings2004transients}
A.~Hastings, Transients: the key to long-term ecological understanding?, Tren.
  Ecol. Evol. 19~(1) (2004) 39--45.

\bibitem{birkhoff1969ordinary}
G.~Birkhoff, G.~Rota, Ordinary differential equations, Wiley, New York, 1969.

\bibitem{seneta1981non}
E.~Seneta, Non-negative matrices and Markov chains, 2nd Edition, Springer, New
  York, 1981.

\end{thebibliography}
\end{document}